\documentclass[11pt,a4]{article}

\usepackage{pstricks, pst-coil, pst-node, pst-tree, multido}
\usepackage{enumerate}
\usepackage{verbatim}
\usepackage{latexsym} 
\usepackage{theorem}
\usepackage{graphics}
\usepackage{graphicx}
\usepackage{amsmath}
\usepackage{amssymb}
\usepackage{clrscode} 
\usepackage{longtable}
\newtheorem{defeng}{Definition}[section]
\newtheorem{theorem}[defeng]{Theorem}

\newtheorem{lemma}[defeng]{Lemma}

\newtheorem{corollary}[defeng]{Corollary}
{\theorembodyfont{\rmfamily} }
{\theorembodyfont{\rmfamily} }
{\theorembodyfont{\rmfamily} }
{\theoremstyle{break}\theorembodyfont{\rmfamily} }
{\theoremstyle{break}\theorembodyfont{\rmfamily} }

\newcounter{claim}

\newenvironment{proof}[1][]%
 {\noindent {\setcounter{claim}{0}\sc proof ---
   }{#1}{}}{\hfill$\Box$\vspace{2ex}} 

{\refstepcounter{claim}\vspace{1ex}\noindent{(\it\arabic{claim}){#1}{}}\it}{\vspace{1ex}}

	{\noindent {}{#1}{}}{ This proves~(\arabic{claim}).\vspace{1ex}}

\newcommand{\sm}{\setminus} 

\DeclareMathOperator{\Mark}{Mark}
\DeclareMathOperator{\Move}{Move}

\DeclareMathOperator{\Explore}{Explore}

\usepackage{ifpdf}

\ifpdf
\fi

\begin{document}

\title{Detecting 2-joins faster}

\author{Pierre Charbit\thanks{Universit\'e Paris 7, LIAFA, Case 7014,
    75205 Paris Cedex 13, France.  E-mail: pierre.charbit@liafa.jussieu.fr.}~,~Michel
  Habib\thanks{Universit\'e Paris 7, LIAFA, Project team Inria :
    Gang, Case 7014, 75205 Paris Cedex 13, France.  E-mail:
    michel.habib@liafa.jussieu.fr.}~, \\Nicolas Trotignon\thanks{CNRS,
    LIP ENS de Lyon, INRIA, Universit\'e de Lyon, 15 parvis Ren\'e
    Descartes	BP 7000	69342 Lyon cedex 07	France. Email: nicolas.trotignon@ens-lyon.fr.} ~and Kristina
  Vu\v{s}kovi\'c\thanks{School of Computing, University of Leeds,
    Leeds LS2 9JT, UK and Faculty of Computer Science, Union
    University, Knez Mihailova 6/VI, 11000 Belgrade, Serbia. E-mail:
    k.vuskovic@leeds.ac.uk.  Partially supported by EPSRC grant
    EP/H021426/1 and Serbian Ministry of Education and Science grants
    174033 and III44006.\newline The first and third authors are
    supported by \emph{Agence Nationale de la Recherche} under
    reference \textsc{anr 10 jcjc 0204 01}. \newline The four authors
    are also supported by PHC Pavle Savi\'c grant, jointly awarded by
    EGIDE, an agency of the French Minist\`ere des Affaires
    \'etrang\`eres et europ\'eennes, and Serbian Ministry for Science
    and Technological Development.}}

\date{September 28, 2012}

\maketitle

\begin{abstract}
  2-joins are edge cutsets that naturally appear in the decomposition
  of several classes of graphs closed under taking induced subgraphs,
  such as balanced bipartite graphs, even-hole-free graphs, perfect
  graphs and claw-free graphs.  Their detection is needed in several
  algorithms, and is the slowest step for some of them.  The classical
  method to detect a 2-join takes $O(n^3m)$ time where $n$ is the
  number of vertices of the input graph and $m$ the number of its
  edges.  To detect \emph{non-path} 2-joins (special kinds of 2-joins
  that are needed in all of the known algorithms that use 2-joins), the
  fastest known method takes time $O(n^4m)$.  Here, we give an
  $O(n^2m)$-time algorithm for both of these problems.  A consequence is a
  speed up of several known algorithms.
\end{abstract}


\section{Introduction}

A partition $(X_1, X_2)$ of the vertex-set of a graph $G$ is a
\emph{2-join} if for $i=1,2$, there exist disjoint non-empty $A_i, B_i
\subseteq X_i$ satisfying the following:

\begin{itemize} 
\item
  every vertex of $A_1$ is adjacent to every vertex of $A_2$,  every
  vertex of $B_1$ is adjacent to every vertex of $B_2$, and
  there are no other edges between $X_1$ and $X_2$;
\item 
  for $i=1,2$, $|X_i| \geq 3$. 
\end{itemize}

Sets $X_1$ and $X_2$ are the two \emph{sides} of the 2-join.  
We say that $(X_1, X_2,
A_1, B_1, A_2, B_2)$ is a \emph{split} of a 2-join $(X_1, X_2)$. 
For $i= 1, 2$, we will denote by $C_i$ the set $X_i \setminus (A_i
\cup B_i)$.

The 2-join was first introduced by Cornu\'ejols and Cunningham in
\cite{cornuejols.cunningham:2join} in the context of studying
composition operations that preserve perfection.  It is a
generalization of an edge cutset known as a 1-join (or \emph{join} or
\emph{split decomposition}) introduced by Cunningham and Edmonds in
\cite{edcu:dec}.  A partition $(X_1,X_2)$ of the vertex set
of a graph $G$ is a \emph{1-join} if for $i=1,2$, there exists
non-empty $A_i \subseteq X_i$ satisfying the following:

\begin{itemize} 
\item
  every vertex of $A_1$ is adjacent to every vertex of $A_2$, and
  there are no other edges between $X_1$ and $X_2$;
\item 
  for $i=1,2$, $|X_i| \geq 2$. 
\end{itemize}

2-Joins ended up playing a key role in structural characterizations of
several complex classes of graphs closed under taking induced
subgraphs, and construction of polynomial time recognition and
optimization algorithms associated with these classes.  2-Joins are
used in decomposition theorems for balanced bipartite graphs that
correspond to balanced $0,1$ matrices \cite{ccr} as well as 
balanced $0,\pm 1$
matrices \cite{cckv-bal}, even-hole-free graphs \cite{cckv-ehf1, dsv},
odd-hole-free graphs \cite{ccv-ohf}, square-free Berge graphs
\cite{conforti.c.v:square}, Berge graphs in general
\cite{spgt,chudnovsky:trigraphs,nicolas:bsp} and claw-free graphs
\cite{chudnovsky.seymour:ClawFree}.  The decomposition theorem in
\cite{spgt} famously proved the Strong Perfect Graph Conjecture.

Decomposition based polynomial time recognition algorithms, that use
2-joins, are constructed for balanced $0, \pm 1$ matrices
\cite{cckv-bal}, even-hole-free graphs \cite{cckv-ehf2, dsv} and Berge
graphs with no balanced skew partition \cite{nicolas:bsp}.  2-Joins
are also used in \cite{nicolas.kristina:2-join} for solving the
following combinatorial optimization problems in polynomial time:
finding a maximum weighted clique, a maximum weighted stable set and
an optimal coloring for Berge graphs with no balanced skew partition
and no homogeneous pairs, and finding a maximum weighted stable set
for even-hole-free graphs with no star cutset.

Detecting a 2-join in a graph obviously reduces to detecting a 2-join
in a connected graph, so input graphs of our algorithms may be assumed
to be connected.  We denote with $n$ the number of vertices of an
input graph $G$, and with $m$ the number of edges in $G$.  In
\cite{cornuejols.cunningham:2join} an ${O}(n^3m)$ algorithm for
finding a 2-join in a graph $G$ (or detecting that the graph does not
have one) is given.  The algorithm is based on a set of forcing rules
that for a given pair of edges $a_1a_2$ and $b_1b_2$ decides, in time
${O} (n^2)$, whether there exists, a 2-join with split
$(X_1,X_2,A_1,B_1,A_2,B_2)$ such that for $i=1,2$, $a_i \in A_i$ and
$b_i \in B_i$, and finds it if it does.  In Section~\ref{sec:force},
we describe a new method to achieve the same goal slightly faster, in
time $O(n+m)$.

It is observed in \cite{cornuejols.cunningham:2join} that since for
any spanning tree $T$ of $G$, any 2-join $(X_1,X_2)$ must contain an
edge of $T$ that is between $X_1$ and $X_2$, then to find a 2-join
in a graph, one needs to check ${O}(nm)$ pairs of edges $a_1a_2$ and
$b_1b_2$, giving the total running time of ${O}(n^3m)$ for finding a
2-join.  In Section~\ref{sec:univ}, we show that actually one only
needs to check ${O}(n^2)$ pairs of edges, reducing the running time of
finding a 2-join to ${O}(n^2m)$.

All the 2-joins whose detection is needed for the algorithms mentioned
above in fact have an additional crucial property: they are
\emph{non-path} 2-joins.  A 2-join is said to be a \emph{path 2-join} if it
has a split $(X_1, X_2, A_1, B_1, A_2, B_2)$ such that for some $i\in
\{1, 2\}$, $G[X_i]$ is a path with an end in $A_i$, an end in $B_i$
and interior in $C_i$.  In this case $X_i$ is said to be a
\emph{path-side} of this 2-join.  A \emph{non-path 2-join} is a 2-join
that is not a path 2-join.  In \cite{cckv-ehf2} it is observed that by
applying the 2-join detection algorithm ${O}(n)$ times one can find a
non-path 2-join if there is one.  In Section~\ref{sec:nonpath} we show
that in fact a constant number of calls to the algorithm for 2-join is
needed, so that non-path 2-joins can also be detected in
${O}(n^2m)$-time.

In inductive arguments or algorithms that use cutsets, i.e.\
decomposition theorems, one needs the concept of the {\em blocks of
  decomposition}, by which a graph is decomposed into ``simpler''
graphs. Blocks of decomposition of a graph $G$ with respect to a
2-join with split $(X_1,X_2,A_1,B_1,A_2,B_2)$ are graphs $G_1$ and
$G_2$ usually constructed as follows: $G_1$ is obtained from $G$ by
replacing $X_2$ by a {\em marker path} $P_2$ that is a chordless path
from a vertex $a_2$ which is adjacent to all of $A_1$ to a vertex
$b_2$ which is adjacent to all of $B_1$, and whose interior vertices
are all of degree two in $G_1$. Block $G_2$ is obtained similarly by
replacing $X_1$ by a marker path $P_1$. In all of the above mentioned
papers, blocks of decomposition for 2-joins are constructed this way,
where marker paths are of some fixed small length. For example in
\cite{cornuejols.cunningham:2join} they are of length 1, and in the
other papers they are of length at most 6. It is now clear why
non-path 2-joins are a more useful concept when using 2-joins in
algorithms.

In \cite{cornuejols.cunningham:2join} it is claimed that at most $n$
applications of the 2-join detection algorithm are needed to decompose
a graph into irreducible factors, i.e.\ graphs that have no 2-join.
This is true, as shown in \cite{cckv-ehf2}, but in
\cite{cornuejols.cunningham:2join} it is based on a wrong observation
that the 2-join detection algorithm given in
\cite{cornuejols.cunningham:2join} always finds an extreme 2-join,
i.e.\ one whose both blocks of decomposition are irreducible. First of
all it is not true that every graph that has a 2-join, has an extreme
2-join. For example graph $G$ in Figure~\ref{fige2j} has exactly two
2-joins, one is represented with bold lines, and the other is
equivalent to it.  Both of the blocks of decomposition are isomorphic
to graph $H$ (where dotted lines represent paths of arbitrary length,
possibly of length 0), and $H$ has a 2-join whose edges are
represented with bold lines.  So $G$ does not have an extreme 2-join.
Even if a graph had an extreme 2-join the algorithm in
\cite{cornuejols.cunningham:2join} would not necessarily find it.  
On the other hand, 1-joins have a much nicer tree-like structure so that 
there exist fast ($O(m)$ time) algorithms to compute a representation 
of the whole family of 1-joins of a given graph, and in particular 
yield extremal ones. See for example 
\cite{ms:split,dahl:split,cmr:split}.
 \begin{figure}[h!]
  \begin{center}
    \psset{xunit=6.0mm,yunit=6.0mm,radius=0.1,labelsep=0.1}

    \def\uputnode(#1,#2)#3#4{\Cnode(#1,#2){#3}\nput{ 90}{#3}{\small$#4$}}
    \def\dputnode(#1,#2)#3#4{\Cnode(#1,#2){#3}\nput{270}{#3}{\small$#4$}}
    \def\lputnode(#1,#2)#3#4{\Cnode(#1,#2){#3}\nput{180}{#3}{\small$#4$}}
    \def\rputnode(#1,#2)#3#4{\Cnode(#1,#2){#3}\nput{0}{#3}{\small$#4$}}
    \def\sputnode(#1,#2)#3#4{\Cnode(#1,#2){#3}\nput{45}{#3}{\small$#4$}}
    \begin{pspicture}(15,9)

      \uputnode(5,1){b2'}{}
      \uputnode(3,1){b1'}{}
 
      \uputnode(2,2){w}{}
      \lputnode(6,2){z}{}
 
      \sputnode(1,3){w1}{}
      \rputnode(3,3){b1}{}
      \dputnode(5,3){b2}{}
      \dputnode(7,3){z1}{}
      \dputnode(1,5){x1}{}

      \dputnode(3,5){a1'}{}
      \uputnode(5,5){a2'}{}
      \uputnode(7,5){y1}{}

      \lputnode(2,6){x}{}
      \sputnode(6,6){y}{}

      \rputnode(3,7){a1}{}
      \dputnode(5,7){a2}{}

      \rput(4,8){\rnode{g}{$G$}}

      \dputnode(11,1){B1'}{}
      \uputnode(10,2){W}{}
      \sputnode(9,3){W1}{}
      \rputnode(11,3){B1}{}
      \dputnode(13,3){B2}{}
      \dputnode(9,5){X1}{}

      \dputnode(11,5){A1'}{}
      \uputnode(15,5){Y1}{}
      \lputnode(10,6){X}{}
      \sputnode(14,6){Y}{}
      \rputnode(11,7){A1}{}
      \dputnode(13,7){A2}{}

      \rput(12,8){\rnode{h}{$H$}}

    \ncline{x}{x1}
      \ncline{x1}{w1}
      \ncline{w}{w1}
      \ncline{w}{b1'}
      \ncline{w1}{a1'}
      \ncline{x1}{b1}
      \ncline{a1'}{b1}

      \ncline{y}{a2}
      \ncline{y}{y1}
      \ncline{y1}{z1}
      \ncline{z}{z1}
      \ncline{z}{b2'}
      \ncline{b2}{y1}
      \ncline{a2'}{z1}
      \ncline{a2'}{b2}

      \ncline{b1'}{w}
      \ncline{x}{a1}  

      \ncline{X}{X1}
      \ncline{X}{A1}
      \ncline{W}{W1}
      \ncline{W}{B1'}
      \ncline{A2}{A1}
      \ncline{A2}{A1'}
      \ncline{B2}{B1}
      \ncline{B2}{B1'}

      \psset{linewidth=0.8mm}

      \ncline{a2}{a1}
      \ncline{a2'}{a1}
      \ncline{a2}{a1'}
      \ncline{a2'}{a1'}
      \ncline{b2}{b1}
      \ncline{b2'}{b1}
      \ncline{b2}{b1'}
      \ncline{b1'}{b2'}

      \ncline{Y}{Y1}
      \ncline{A1'}{B1}
      \ncline{W1}{A1'}
      \ncline{X1}{W1}
      \ncline{X1}{B1}

      \psset{linestyle=dotted}
      \ncline{A2}{Y}
      \ncline{B2}{Y1}      

    \end{pspicture}
\end{center}
\caption{A graph $G$ with no extreme 2-join\label{fige2j}}
\end{figure}
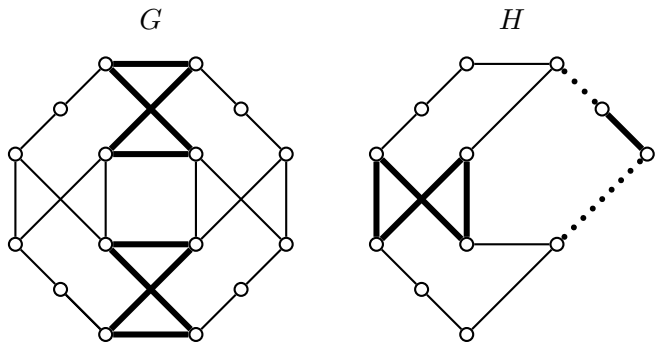

For the optimization algorithms in~\cite{nicolas.kristina:2-join}, it
is in fact essential that these extreme non-path 2-joins are used,
which is potentially a problem since as shown above, a graph with a
2-join may fail to have an extreme 2-join.  Fortunately, graphs studied
in~\cite{nicolas.kristina:2-join} have no star cutset, where a
\emph{star cutset} is any set $S\subseteq V(G)$ such that $G \setminus
S$ is disconnected and for some $x \in S$, $x$ is adjacent to all
vertices of $S \setminus \{ x \}$.  And as shown in
\cite{nicolas.kristina:2-join}, if a graph with no star star cutset
has a non-path 2-join, then it has an extreme non-path 2-join.  In
Section~\ref{sec:extreme} we show how to find an extreme non-path
2-join in time $O(n^3m)$ in graphs that have no star cutset.  It is
in fact interesting that for \emph{all} known algorithms that use
2-join detection (see the list in Section~\ref{sec:consequence}), one
actually needs to look for a non-path 2-join in graphs that do not have
star cutsets.  This remark could perhaps lead to further speed ups.

In Section~\ref{sec:consequence} we survey the consequences of our
work for the running time of several algorithms that use 2-joins.


\section{Finding a 2-join compatible with a 4-tuple}
\label{sec:force}

A 4-tuple $(a_1, a_2, b_1, b_2)$ of vertices from a graph $G=(V,E)$ is
\emph{proper} if:
\begin{itemize}
\item $a_1$, $b_1$, $a_2$, $b_2$ are pairwise distinct;
\item $a_1a_2, b_1b_2 \in E$;
\item $a_1b_2, b_1a_2 \notin E$.
\end{itemize}

It is \emph{compatible} with a 2-join $(X_1, X_2)$ of $G$ if $a_1, b_1
\in X_1$ and $a_2, b_2 \in X_2$.  Note that when $(X_1, X_2)$ has
split $(X_1, X_2, A_1, B_1, A_2, B_2)$ then for any $a_1 \in A_1$,
$a_2 \in A_2$, $b_1 \in B_1$ and $b_2 \in B_2$, the 4-tuple $(a_1,a_2,
b_1, b_2)$ is proper and compatible with $(X_1, X_2)$; and any proper
4-tuple $(a_1,a_2,b_1,b_2)$ that is compatible with a 2-join
$(X_1,X_2)$ is such that for $i=1,2$, either $a_i \in A_i$ and $b_i
\in B_i$, or $a_i \in B_i$ and $b_i \in A_i$.

In~\cite{cornuejols.cunningham:2join}, Cornu\'ejols and Cunningham
describe a set of forcing rules that output a 2-join of a graph
compatible with a given 4-tuple, if there exists one.  The method is
implemented in time $O(n^2)$.  Here, we propose something slightly
faster for the same task.

\begin{lemma}\label{l:forcing}
  Let $G$ be a graph and $Z = (a_1, a_2, b_1, b_2)$ a proper 4-tuple
  of $G$.  There is an $O(n+m)$-time algorithm that given a set $S_0
  \subseteq V(G)$ of size at least $3$ such that $\{a_1, b_1, a_2,
  b_2\} \cap S_0 = \{a_1, b_1\}$ (resp.\ $\{a_1, b_1, a_2, b_2\} \cap
  S_0 = \{a_2, b_2\}$) outputs a 2-join with a split $(X_1, X_2, A_1,
  B_1, A_2, B_2)$, compatible with $Z$ and such that $a_1 \in A_1$,
  $a_2 \in A_2$, $b_1 \in B_1$, $b_2 \in B_2$ and $S_0 \subseteq X_1$
  (resp.\ $S_0 \subseteq X_2$), if there exists such a 2-join.

  Moreover, $X_1$ (resp.\ $X_2$) is minimal with respect to this
  property, meaning that any 2-join with split $(X'_1, X'_2, A'_1,
  B'_1, A'_2, B'_2)$ satisfying these properties is such that $X_1
  \subseteq X'_1$ (resp.\ $X_2 \subseteq X'_2$).
\end{lemma}

\begin{proof}
\begin{table}\label{algoM}
\begin{description}
\item{\bf Input:} $S_0$ a  set of vertices of a graph $G$ such that:
  $|S_0| \geq 3$ and four vertices $a_1$, $b_1$, $a_2$, $b_2$ pairwise
  distinct with: $a_{1}, b_{1} \in S_0$, $a_{2}, b_{2} \notin S_0$, 
  $a_{1}a_{2}, b_{1}b_{2}\in E$, and $a_{1}b_{2}, b_{1}a_{2}\notin E$.

\item{\bf Initialization:}

$S \leftarrow S_0$; $T \leftarrow V(G) \sm S_0$; 
$A \leftarrow N(a_{1})\cap T$;  $B \leftarrow N(b_{1})\cap T$;

\textbf{If} $A \cap B \neq \emptyset$  \textbf{then} $\Move(A \cap B)$;

Vertices $a_{1}, b_{1}, a_{2}, b_{2}$ are left unmarked. For the other vertices of $G$:

\rule{1em}{0ex}$\Mark(x) \leftarrow \alpha.\beta$ for every vertex $x \in  N(a_{2})\cap N(b_{2}) $;

\rule{1em}{0ex}$\Mark(x) \leftarrow \alpha$ for every vertex $x \in  N(a_{2}) \sm N(b_{2}) $;

\rule{1em}{0ex}$\Mark(x) \leftarrow \beta$ for every vertex $x \in  N(b_{2}) \sm N(a_{2}) $;

\rule{1em}{0ex}Every other vertex of $G$ is marked by $\varepsilon$;

\rule{1em}{0ex}{\it Note that a vertex can be unmarked, 
or marked by $\varepsilon, \alpha, \beta~or~\alpha.\beta$. }


\item{\bf Main loop:}

\textbf{While} there exists a vertex $x \in S$ marked

\textbf{Do} $\Explore(x);$ unmark vertex $x$;

\item{\bf Function Explore(x):}

Case on the value of  $\Mark(x)$:

\rule{1em}{0ex}\parbox{15cm}{\textbf{If} 
$\Mark(x)=\alpha.\beta$ \textbf{then} STOP;\\
\rule{1em}{0ex}{OUTPUT : {\it   No 2-join $(S,T)$ with $S_{0}\subset S$ 
is compatible with the 4-tuple.}}}

\rule{1em}{0ex}\textbf{If} $\Mark(x)=\alpha$  \textbf{then}  
$\Move( A \Delta (N(x)\cap T))$;

\rule{1em}{0ex}\textbf{If} $\Mark(x)=\beta$ \textbf{then}  
$\Move( B \Delta (N(x)\cap T))$;

\rule{1em}{0ex}\textbf{If} $\Mark(x)=\varepsilon$ \textbf{then} 
$\Move(N(x)\cap T)$;

\item{\bf Function Move(Y):}

{\it This function just moves a subset $Y \subset T$ from $T$ to $S$. }

$S \leftarrow S \cup Y$;  $ A\leftarrow A \sm Y$;   $ B\leftarrow B \sm Y$;  $ T\leftarrow T \sm Y$;
\end{description}

\caption{Procedure used in Lemma~\ref{l:forcing}\label{algoM}}
\end{table}
We use the procedure described in Table~\ref{algoM}.  The following
properties are easily checked to be invariant during all the
execution of the procedure (meaning that they are satisfied after each
call to Explore):

\begin{itemize}
\item $S$ and $T$ form a partition of $V(G)$, $S_0\subseteq  S$ and
  $a_2, b_2 \in T$.
\item
All unmarked vertices belonging to   $S\cap N(a_{2})$ have the
same neighborhood in $T$, namely $A$. 

\item All  unmarked  vertices  belonging to   $S\cap N(b_{2})$ have
  the same neighborhood in $T$, namely $B$.

\item All unmarked  vertices belonging to $S$  which do not see $a_{2}$ 
nor $b_{2}$
  have the same  neighborhood  in $T$, namely $\emptyset$.

\item For every 2-join $(X_1,X_2)$ such that $S_0 \subseteq X_1$ and
  $a_2,b_2 \in X_2$, we have that $S \subseteq X_1$ and $X_2 \subseteq
  T$.
\end{itemize}

Since all moves from $T$ to $S$ are necessary (this comes from the last
item), if we find a vertex marked $\alpha.\beta$ in $S$ then no 
desired 2-join
exists.  If the process does not stop because of a vertex marked
$\alpha.\beta$ then all vertices of $S$ have been explored and
therefore are unmarked.  So, if $|T|\geq 3$, at the end, $(S, T)$, is
a 2-join compatible with $Z$: $X_1 = S$, $X_2 = T$, $A_1 = S \cap
N(a_2)$, $B_1 = S \cap N(b_2)$, $A_2 = T \cap N(a_1)$, $B_2 = T \cap
N(b_1)$.  Since all moves from $T$ to $S$ are necessary, the 2-join is
minimal as claimed (this also implies that if $|T| \leq 2$, then
no desired 2-join exists).

\noindent\textbf{Complexity Issues:} The neighborhood of a vertex in
$S$ is considered at most once.  So, globally, the process requires
$O(n+m)$ time.
\end{proof}

\begin{corollary}
  \label{usualtrick}
  There is an $O(n+m)$ algorithm whose input is a graph $G$ together
  with a proper 4-tuple $Z$ of vertices and whose output is a 2-join of
  $G$ compatible with $Z$ if such a 2-join exists.
\end{corollary}
\begin{proof}
  We suppose $|V(G)|\geq 6$ for otherwise no 2-join exists.  Suppose
  $Z= (a_1, a_2, b_1, b_2)$.  Take any vertex $u$ of $G \setminus \{
  a_1,a_2,b_1,b_2\}$ and apply Lemma~\ref{l:forcing} to $S_0=\{
  a_1, b_1, u\}$ and then to $S_0=\{ a_2, b_2, u \}$.  Since for any
  2-join $(X_{1},X_{2})$ compatible with $Z$, either $u \in X_1$ or $u
  \in X_2$, this method detects a 2-join compatible with $Z$ if there
  is one.
\end{proof}

\section{Computing a small universal set}
\label{sec:univ}

A set $U$ of proper 4-tuples of vertices of a graph $G$ is
\emph{universal} if for every 2-join $(X_1, X_2)$ of $G$, at least one
4-tuple from $U$ is compatible with $(X_1, X_2)$.  Note that for all
graphs, there exists a universal set: the set of all proper 4-tuples
of vertices.  Note that if a graph has no 2-join then any set of
proper 4-tuple of vertices, including the empty set, is universal.

To detect 2-joins, it suffices to consider a universal set $U$, and to
apply Corollary~\ref{usualtrick} for all 4-tuples $Z = (a_1, a_2, b_1,
b_2)$ in $U$.  This gives a naive $O(n^4m)$ time algorithm (by
considering the universal set of all proper 4-tuples) and an
${O}(nm^2)$ (that was originally $O(n^3m)$) algorithm for finding a
2-join as described in~\cite{cornuejols.cunningham:2join}, by
considering a universal set of size $O(nm)$ as explained in the
introduction.  We now show how to compute a universal set $U$ of
proper 4-tuples of $G$, of size $O(n^2)$, for any connected graph $G$,
resulting in an $O(n^2m)$ algorithm for finding a $2$-join.

We first review some well known facts about breadth first search trees.
When $T$ is a tree and $u, v$ are vertices of $T$, we denote by $uTv$
the unique path of $T$ between $u$ and $v$.  For a graph $G$ and
vertices $u$ and $v$ of $G$ we denote by $d_G(u, v)$ the distance
between $u$ and $v$ in $G$.  A \emph{BFS-tree} of a connected graph
$G$ is any couple $(T, r)$ where $r$ is a vertex of $G$ and $T$ is a
spanning tree of $G$ such that for all vertices $v\in V(G)$ we have
$d_T(r, v) = d_G(r, v)$.  We say that $r$ is the \emph{root} of $T$.
It is a well known result that for any connected graph $G$ and any
vertex $r$, there exists a BFS-tree $(T, r)$.

Once a BFS tree $(T, r)$ of a graph $G$ is given, we use the following
standard terminology.  The \emph{level} of a vertex $v$ is $l(v) =
d_G(r, v) = d_T(r, v)$.  For any vertex $v\neq r$, there exists a
unique vertex $u$ such that $uv \in E(T)$ and $l(u) = l(v)-1$.  We say
that $u$ is the \emph{parent} of $v$ and $v$ is a \emph{child} of $u$.
We denote by $p(v)$ the parent of $v$. The vertices of $rTv$ are the
\emph{ancestors} of $v$.  If $v$ is a vertex of $rTu$ then $u$ is a
\emph{descendant} of $v$.

A well known linear-time algorithm, named \emph{BFS}, computes a
BFS-tree $(T, r)$ of any connected graph $G$ for any vertex $r$.  The
algorithm provides as an output the tree together with a $O(n)$-time
routine that allows to compute the parent and all the ancestors of any
non-root vertex, and the children and all the descendants of any
vertex.  For the implementation, see for instance~\cite{gibbons:agt}.

Consider the following method for computing a set $U$ of 4-tuples.

\begin{enumerate}
\item Start with $U = \emptyset$.
\item Choose a vertex $r$ and run BFS to obtain a BFS-tree $(T, r)$.
\item\label{i:pairs} Add to $U$ all proper 4-tuples $(a_1, a_2, b_1,
  b_2)$ such that $a_1a_2, b_1b_2 \in E(T)$.
\item\label{i:anscestor} For all pairs of vertices $u$ and $v$ of $G$
  such that $l(u) \geq 2$ and $l(v) \geq 1$ do the following:
  \\
  Compute the set $D_u$ of all descendant of $u$ (note that $u\in
  D_u$).
  \\
  If there exists an edge $a_1v$ with $a_1\in D_u$, pick any such edge
  and add $(a_1, v, p(u), p(p(u)))$ and $(p(p(u)), p(u),  v, a_1)$ to
  $U$ (when they are proper).
\end{enumerate}

\begin{lemma}\label{lu}
  A connected graph $G$ admits a universal set of 
proper 4-tuples of $G$,
of size $O(n^2)$.  Such a set can be found in time $O(n^3)$.
\end{lemma}

\begin{proof}
  We use the method above.  It obviously computes a set $U$, of size
  $O(n^2)$, made of proper 4-tuples of $G$.  The complexity of this
  computation is dominated by step (iv). In this step for ${O}(n^2)$
  pairs of vertices $u$ and $v$, we first compute $D_u$, which can be
  done in time ${O}(n)$ (since we already have the tree $T$), and then
  we check whether $v$ is adjacent to a vertex of $D_u$, which again
  can be done in time ${O}(n)$. So the total complexity is ${O}(n^3)$.
  It remains to prove that $U$ is universal.  Let $(X_1, X_2, A_1,
  B_1, A_2, B_2)$ be a split of a 2-join of $G$.

  If $T$ contains an edge $a_1a_2$ between $A_1$ and $A_2$, and an
  edge $b_1b_2$ between $B_1$ and $B_2$, then, in step~\ref{i:pairs},
  the proper 4-tuple $(a_1, a_2, b_1, b_2)$ is compatible with $(X_1,
  X_2)$ and added to $U$.  So, from here on, up to a relabeling, we
  assume that $T$ contains no edge between $A_1$ and $A_2$.

  Suppose first $r\in X_2$.  Pick any vertex $a$ in $A_1$.  Since $G$
  is connected, $a\in V(T)$.  So, the ancestors of $a$ form a shortest
  path $P$ from $r$ to $a$.  Path $P$ must have an edge $b_1b_2$ where
  $b_1 \in B_1$ and $b_2 \in B_2$.  Note that $P$ is chordless, and
hence  
  $b_1$ is the unique vertex of $P$ in $B_1$, and $b_2$ the unique
  vertex of $P$ in $B_2$.  Let $u$ be the vertex of $P$ such that $u,
  b_1, b_2$ are consecutive along $P$.  Note that possibly, $u=a$.
  So, $b_1 = p(u)$ and $b_2 = p(p(u))$.

  We claim that $D_u$ is included in $X_1$.  Indeed, because of $b_2$,
  every vertex $x$ in $B_1$ satisfies $l(x) \leq l(b_2) +1$.  And any
  descendant $y$ of $u$ satisfies $l(y) \geq l(u) = l(b_2) + 2$.  So, no
  descendant of $u$ is in $B_1$.  Since $T$ contains no edge between
  $A_1$ and $A_2$, no descendant of $u$ can be in $X_2$.  

  Let $v$ be any vertex of $A_2$.  Note that since no edge between
  $A_1$ and $A_2$ is in $T$, $l(v) \geq 1$.  At some point in
  Step~\ref{i:anscestor}, the algorithm considers $u$ and $v$.  Since
  $a\in D_u$, there exists an edge between $D_u$ and $v$, and any such
  edge $a_1v$ must be between $A_1$ and $A_2$ because $D_u \subseteq
  X_1$.  So, $(a_1, v, b_1, b_2) = (a_1, v, p(u), p(p(u)))$ is proper,
  compatible with $(X_1, X_2)$ and added to $U$.

  When $r\in X_1$, we find similarly that a proper 4-tuple 
$(p(p(u)),p(u),v,a_{1})$ is added to $U$.
\end{proof}

\begin{theorem}
  There is an $O(n^2m)$-time algorithm that outputs a 2-join of an input
  graph, or  certifies that no such 2-join exists. 
\end{theorem}

\begin{proof}
  By Lemma~\ref{lu}, compute an universal set $U$ of $O(n^2)$ proper
  4-tuples in time $O(n^3)$.  Apply Corollary~\ref{usualtrick} to
  each 4-tuple in $U$.  This leads to an  $O(n^2m)$-time algorithm.
  In case of failure, $U$ is a certificate.
\end{proof}

This algorithm is quite brute force and in the worst case, many 
computations are repeated many times. In fact, we do not know any 
construction of instances for which the worst case is actually achieved. 
So, a faster implementation might exist.


\section{Detecting non-path 2-joins}
\label{sec:nonpath}
The purpose of this section is to prove the following theorem.

\begin{theorem}
  There is an $O(n^2m)$-time algorithm that outputs a non-path 2-join
  of an input graph, or certifies that no non-path 2-join exists.
\end{theorem}

\begin{proof}
  The idea is similar to the one of the previous section. First we
  compute a universal set $U$ of size $O(n^2)$ by Lemma~\ref{lu}.
  Then, for every 4-tuple $Z=(a_1, a_2, b_1, b_2)$ in $U$, we either
  find a non-path $2$-join compatible with $Z$ or certify that none
  exists.

  Therefore let us fix $Z$ and define a {\em bad path} to be, for
  $i=1,2$, an induced path of $G$ of length at least $2$, from $a_i$
  to $b_i$, avoiding $a_{3-i}$ and $b_{3-i}$, whose interior vertices
  are all of degree $2$.  Note now that a $2$-join is a non-path
  $2$-join if and only if none of the two sides is a bad path.  We
  check now whether there exists a vertex $u$ that is not in any bad
  path of the input graph. This is  easy to do in linear time by
  computing the degrees and searching the graph.

  Suppose first that we find such a vertex $u$.  Then we apply
  Lemma~\ref{l:forcing} to $S_0 = \{a_1, b_1, u\}$ and to $S_0 =
  \{a_2, b_2, u\}$.  We claim that this will detect a non-path 2-join
  compatible with $Z$ if there is one.  Indeed, suppose there is one
  and suppose up to symmetry that $u, a_1, b_1$ are in the same side.
  When we apply Lemma~\ref{l:forcing} to $\{a_1, b_1, u\}$, some
  2-join $(X_1, X_2)$ must be detected.  If it is a path 2-join, then
  the path-side must be $X_2$ because $u$ cannot be in a path-side
  since it is not in any bad path.  But since $X_1$ is minimal in the
  sense of Lemma~\ref{l:forcing}, we see that any 2-join compatible
  with $Z$ and with $a_1, b_1, u$ in the same side must in fact be
  $(X_1, X_2)$, and hence a path 2-join, contradicting our assumption
  (indeed, since $X_2$ is a path, no vertex can be moved from $X_2$ to
  $X_1$).  So $(X_1, X_2)$ is non-path and we output it.  This
  completes the proof when there exists a vertex that is not in any
  bad path.
  
  Now we may assume that every vertex of $G\sm Z$ is in a bad path.
  This means that $G$ is the union of paths from $a_i$ to $b_i$, $i=1,
  2$, all of length at least 2, with interior vertices of degree 2,
  plus the two edges $a_1a_2$ and $b_1b_2$.  Then it is
  straightforward to decide directly whether a non-path 2-join
  compatible with $Z$ exists by just counting: let $k$ be the number
  of bad paths; if $k \leq 2$, or $k=3$ and all of the vertices of $Z$
  are in bad paths, then no non-path 2-join exists; otherwise a non-path
  2-join exists (and is easy to find by putting two bad paths with
  same endvertices on one side and all the other vertices on the
  other side).
\end{proof}


\section{Finding minimally-sided 2-joins}
\label{sec:extreme}

A non-path 2-join $(X_1,X_2)$ of a graph $G$ is {\em minimally-sided}
if for some $i \in \{ 1,2 \}$, the following holds: for every non-path
2-join $(X_1',X_2')$ of $G$, neither $X_1' \subsetneq X_i$ nor $X_2'
\subsetneq X_i$ holds. In this case $X_i$ is said to be a {\em minimal
  side} of this minimally-sided non-path 2-join. Note that any graph
that has a non-path 2-join, also has a minimally-sided non-path
2-join.  A non-path 2-join $(X_1,X_2)$ of a graph $G$ is an {\em
  extreme 2-join} if for some $i \in \{ 1,2 \}$, the block of
decomposition $G_i$ has no non-path 2-join.  Note that this definition
could be sensitive to the precise definition of what we call a block
of decomposition, but we do not need the definition here.

Recall that graphs that have a non-path 2-join do not necessarily have an
extreme 2-join, as shown in Figure \ref{fige2j}.
For the combinatorial optimization algorithms in
\cite{nicolas.kristina:2-join}, extreme 2-joins are needed.  The
graphs in \cite{nicolas.kristina:2-join} have no star
cutsets, and it is shown in \cite{nicolas.kristina:2-join} that in graphs
with no star cutsets
being a minimally-sided 2-join implies being an extreme 2-join.
So, for the needs in \cite{nicolas.kristina:2-join}, it is enough to
detect minimally-sided non-path 2-joins in graphs with no star
cutsets.  

Note that Lemma~\ref{l:forcing} ensures that the 2-joins that we
detect satisfy a minimality condition, so one might think that the
algorithm in Section~\ref{sec:nonpath} detects a minimally-sided
non-path 2-join.  As far as we can see, this is not the case.  Indeed,
suppose for instance that luckily, the first call to Lemma~\ref{l:forcing}
with $S_0 = \{a_1, b_1, u\}$ gives a non-path 2-join $(X_1, X_2)$.
Then, Lemma~\ref{l:forcing} ensures only that $X_1$ is minimal among
all 2-joins with $a_1, b_1, u$ in the same side, not among all
possible 2-joins compatible with $(a_{1},a_{2},b_{1},b_{2})$.  So, 
to detect minimally-sided 2-joins, a method is
to try all possible vertices $u$.  Below, we show that this works for
non-path 2-joins.  We use Lemma~4.2 from~\cite{nicolas.kristina:2-join}.

\begin{lemma}[\cite{nicolas.kristina:2-join}]\label{nk}
  Let $G$ be a connected graph with no star cutset, and let
  $(X_1,X_2,A_1,B_1,A_2,B_2)$ be a split of a 2-join of $G$.  If
  $(X_1,X_2)$ is a minimally-sided non-path 2-join, with $X_i$ being a
  minimal side, then $|A_i| \geq 2$ and $|B_i|\geq 2$.
\end{lemma}

\begin{theorem}
  \label{th:nostarMinnonPath}
  There is an $O(n^3m)$-time algorithm that outputs a minimally-sided
  non-path 2-join of an input graph with no star cutset, or certifies
  that this graph has no non-path 2-join.
\end{theorem}

\begin{proof}
  We compute a universal set $U$ of size $O(n^2)$ by Lemma~\ref{lu}.
  Then, for all 4-tuple $(a_1, a_2, b_1, b_2)$ in $U$, and for all
  vertices $u$, we apply Lemma~\ref{l:forcing} for $S_0= \{a_1, b_1,
  u\}$ and for $S_0= \{a_2, b_2, u\}$.  This will detect a
  minimally-sided non-path 2-join if there is one. Indeed, suppose
  there is one, with a split $(X_1, X_2, A_1, B_1, A_2, B_2)$ such
  that for $i=1, 2$ we have $a_i\in A_i$, $b_i\in B_i$ and
  $(a_{1},a_{2},b_{1},b_{2}) \in U$.  Then we may assume that up to
  symmetry, $X_1$ is the minimal side.  By Lemma~\ref{nk}, \ $A_1\geq
  2$.  So, for some chosen vertex $u\in A_1\sm \{a_1\}$, we have $S_0=
  \{a_1, b_1, u\} \subseteq X_1$, so Lemma \ref{l:forcing} applied to
  $S_0$ yields a 2-join $(X'_1, X'_2)$ such that $X'_1$ is minimal
  among all the 2-joins compatible with $(a_{1}, a_{2}, b_{1}, b_{2})$
  with $\{a_1, b_1, u\}$ in the same side, so a minimally-sided
  2-join. Note that since $u$ and $a_{1}$ are both adjacent to
  $a_{2}$, $X'_{1}$ cannot be a path side of the $2$-join
  $(X'_{1},X'_{2})$. By the minimality of $X_1$ we have $X_1 \subseteq
  X'_1$, and by Lemma~\ref{l:forcing} we have $X'_1 \subseteq X_1$.
  It follows that $X'_{1}=X_{1}$ and $X'_{2}=X_{2}$.
  
  Therefore by running the procedure of Lemma \ref{l:forcing} for all
  4-tuples $(a_1, a_2, b_1, b_2)$ in $U$, and all vertices $u$, and by
  throwing out every path $2$-join, we get a list of $O(n^3)$ non-path
  $2$-joins that must contain every minimally-sided non-path $2$-join
  of the graph. It suffices now to go through the list and pick a
  $2$-join with fewest number of nodes on one side. That is a
  minimally-sided non-path $2$-join.
\end{proof}

The following algorithms are potentially useful although so far, they
are not needed in any algorithm we are aware of.  We do not exclude
star cutsets anymore, at the expense of a slower running time.  A
2-join $(X_1,X_2)$ of a graph $G$ is {\em minimally-sided} if for some
$i \in \{ 1,2 \}$, the following holds: for every 2-join $(X_1',X_2')$
of $G$ , neither $X_1' \subsetneq X_i$ nor $X_2' \subsetneq X_i$
holds. Note that it is the same definition as for non-path 2-joins,
except that the condition ``non-path'' is omitted.

\begin{theorem}
  There is an $O(n^3m)$-time algorithm that outputs a minimally-sided
  2-join of an input graph or certifies that this graph has no 2-join.
\end{theorem}

\begin{proof}
  The same algorithm as in Theorem~\ref{th:nostarMinnonPath} works.
  Since we do not look for a non-path 2-join, we do not need to use
  Lemma~\ref{nk} and we do not throw out every path $2$-join we
  obtain.
\end{proof}

\begin{theorem}
  There is an $O(n^4m)$-time algorithm that outputs a minimally-sided
  non-path 2-join of an input graph, or certifies that this graph has
  no non-path 2-join.
\end{theorem}

\begin{proof}
  The algorithm from Theorem~\ref{th:nostarMinnonPath} fails, because
  it may happen that a minimally-sided non-path 2-join has its minimal
  side made of the union of two bad paths.  So, we use the same method
  as in Theorem~\ref{th:nostarMinnonPath}, but we check all pairs of
  vertices $u, v$ instead of all vertices $u$, and we apply
  Lemma~\ref{l:forcing} to $\{a_1, b_1, u, v\}$ and $\{a_2, b_2, u,
  v\}$.  Since for any non-path side $X$ of a 2-join, there exist two
  vertices $u, v \in X$ that do not lie on the same bad path, this
  method detects a non-path 2-join when there is one.  We omit further
  details since they are similar to these of
  Theorem~\ref{th:nostarMinnonPath}.
\end{proof}


\section{Consequences}
\label{sec:consequence}

The consequences of finding a non-path 2-join in ${O}(n^2m)$ time, and
finding a minimally-sided non-path 2-join for graphs with no star
cutsets in ${O}(n^3m)$ time, are the following speed-ups of existing
algorithms. Note that the speed-ups are sometimes more than by a
factor of $O(n^2)$. This is because in the algorithms mentioned below
even cruder implementations of non-path 2-join detection are used.

\begin{enumerate}
  \item Detecting balanced skew partitions in Berge graphs in time
    $O(n^5)$ instead of $O(n^9)$ \cite{nicolas:bsp}.
  \item The decomposition based recognition algorithm for Berge graphs
    in \cite{recogberge} is now ${O}(n^{15})$ instead of
    ${O}(n^{18})$, which is not so interesting since the recognition
    algorithm in the same paper that is not based on the decomposition
    method is ${O}(n^9)$. 
  \item\label{i:csc} Finding a maximum weighted clique and a maximum weighted
    stable set in time $O(n^6)$ instead of $O(n^9)$ in Berge graphs
    with no balanced skew partition and no homogeneous pairs, and
    finding an optimal coloring in time $O(n^7)$ instead of
    $O(n^{10})$ for the same class \cite{nicolas.kristina:2-join}.
  \item\label{i:ehf} Finding a maximum weighted stable set in time $O(n^6)$ instead
    of $O(n^9)$ in even-hole-free graphs with no star cutset
    \cite{nicolas.kristina:2-join}.
\end{enumerate} 

As far as we care only for these applications, it is not immediately
usable to try detecting non-path 2-joins faster than $O(n^2m)$, because
$O(n^5)$ is a bottleneck independent from 2-join detection for all the
algorithms mentioned here.  An $O(n^4)$-time algorithm for extreme (or
minimally-sided) non-path 2-joins would allow a speed-up of a factor
$n$ in the algorithms~\ref{i:csc} and \ref{i:ehf}.  We leave this as
an open question.

\end{document}